\documentclass[12pt]{iopart}
\usepackage{iopams,amsthm}
\usepackage{tikz}
\usetikzlibrary{plotmarks}
\usepackage{cite}

\def \dd {\mathrm{d}}
\def \D {\mathrm{D}}
\def \T {\triangle}
\def \calD {\mathcal{D}}
\def \bl {\boldsymbol{\ell}}
\def \bn {\boldsymbol{n}}
\def \bm {\boldsymbol{m}}
\def \btm {\boldsymbol{\tilde{m}}}
\def \bdiff {b_{\mathrm{diff}}}

\newtheorem{thm}{Theorem}
\newtheorem{lemma}{Lemma}
\newtheorem{proposition}{Proposition}
\newtheorem*{definition}{Definition}

\begin{document}

\title{Universal Walker metrics}
\author{S. Hervik$^1$, T. M\'alek$^2$}
\address{$^1$Faculty of Science and Technology,\\     
  University of Stavanger,\\  N-4036 Stavanger, Norway}
\ead{sigbjorn.hervik@uis.no}
\address{$^2$ Institute of Mathematics of the Czech Academy of Sciences,\\ \v{Z}itn\'a 25, 115 67 Prague 1, Czech Republic}   
\ead{malek@math.cas.cz}

\begin{abstract}
We consider four-dimensional spaces of neutral signature and give examples of universal spaces of Walker type. These spaces have no analogue in other signatures in four dimensions and provide with a new class of spaces being universal. 
\end{abstract}


\maketitle

\section{Introduction}

In recent years, various higher-order generalizations of Einstein gravity have attracted attention to researchers \cite{Clifton:2011}. Examples of such theories are dimensionally-reduced higher-dimensional theories; theories involving other fields, or  gravity theories stemming from  Lagrangians being a polynomial curvature invariant of the Riemann curvature tensor and its derivatives; i.e., of the form 
\begin{equation}
  L = L(g_{ab}, R_{abcd}, \nabla_{a_1} R_{bcde}, \dots, \nabla_{a_1 \cdots a_p} R_{bcde}).
  \label{eq:Lagrangian}
\end{equation}
However, the complexity of such theories can be quite intricate and in general there are few  exact solutions found. It is the latter class of theories we will consider here. 

In spite of the difficulty of exact solutions, one may identify a special class of metrics that are vacuum solutions
of Einstein's gravity being immune to any correction term to the Einstein--Hilbert action.
These additional curvature invariants in the Lagrangian give rise to  additional correction
terms in the corresponding field equations which are conserved symmetric rank-2 tensors, $T_{\mu\nu}$.
This fact lead us to the formal definition of \emph{universal metrics} \cite{Coleyetal08}:

\begin{definition}[Universal metrics]
  A metric is $k$-universal if all conserved symmetric rank-2 tensors constructed
  from the metric, the Riemann tensor and its covariant derivatives up to the $k^{\mathrm{th}}$
  order are multiplies of the metric.
  If a metric is $k$-universal for all $k$, then it is simply called \emph{universal}.
\end{definition}
Hence, universal metrics are metrics which obey $T_{\mu\nu}=\lambda g_{\mu\nu}$, for all $T_{\mu\nu}$ being a variation of an action of the form (\ref{eq:Lagrangian}).
Therefore, universal metrics simultaneously solve vacuum field equations of all higher-order theories
of gravity. 
Since the Einstein tensor is an example of such a conserved symmetric rank-2 tensor,
universal metric are thus necessarily Einstein.
Note that, for $d$-dimensional Einstein metrics, it holds
\begin{equation}
  R_{ab} = \frac{R}{d}g_{ab}, \qquad
  R_{abcd} = C_{abcd} + \frac{2R}{d(d-1)} g_{a[c} g_{d]b}
\end{equation}
with $R$ being constant
and consequently
\begin{equation}
  R_{abcd;e} = C_{abcd;e}.
\end{equation}
In other words, a metric is universal, if it is Einstein and all conserved symmetric
rank-2 tensors constructed from the metric, the Weyl tensor and its derivatives
are multiples of the metric.

For Riemannian spaces (i.e., of signature $(+,+,...,+)$) such spaces were studied by Bleecker \cite{Bleecker} and it was found that all such spaces are locally homogeneous isotropy-irreducible spaces. In Lorentzian signature (i.e., of signature $(-,+,...,+)$), spaces which are universal but not locally homogeneous exist; for example, in \cite{HorSte90} it was noticed that Einstein pp-waves are vacuum solutions to any higher-order gravity theory. The Lorentzian case   has later been systematically studied \cite{Coleyetal08,HerPraPra14,univII} and all known Lorentzian examples of universal spacetimes are Kundt. Indeed, in 4 dimensions, it was proven recently that all universal metrics are Kundt \cite{HerPraPra17}. 
Note that the Kundt class is geometrically defined as metrics admitting a non-expanding, non-shearing and non-twisting
null geodetic congruence $\bl$ and the canonical form of such metrics is \cite{Kundt,Coleyetal2003,ColHerPel06,PodolskyZofka2009}
\begin{equation}
  \fl
  \dd s^2 = 2 \,\dd u \,\dd v + 2 H(u, v, x^k) \,\dd u^2 + 2 W_i(u, v, x^k) \,\dd u \,\dd x^i + g_{ij}(u, x^k) \,\dd x^i \,\dd x^j
  \label{eq:Kundt}
\end{equation}
with $\bl = \partial_v$.
In 4 dimensions, the only other possibility is the Neutral case of signature $(-,-,+,+)$ which will be the case we will consider here. Note that in dimensions higher than 4 other possibilities occur.

 We would like to mention one interesting result from \cite{HerPraPra14} which is independent of signature and dimension which will be used later:
\begin{thm}
A universal space is necessarily CSI.
\end{thm}
\noindent CSI spaces are spaces for which all polynomial curvature invariants are constants; see e.g., \cite{ColHerPel06}. This result implies that certain coefficients of the curvature tensors can be chosen to be constants, but they need not be locally homogeneous in the Lorentzian or neutral case.

In this paper we will consider four-dimensional spaces of neutral signature and provide with examples of universal spaces which are not locally homogeneous nor Kundt. Thus these examples have no analogue in Lorentzian or Riemannian signature. The study of universal four-dimensional neutral spaces are part of a larger project aiming to classify all universal spaces regardless of signature. Indeed, in this context, it is noted that all examples of Lorentzian universal metrics have 'seed' metrics which are all Wick-rotated Riemannian universal spaces \cite{wick}. The universal property is invariant under a Wick-rotation and hence, the Wick-rotatable universal metrics seem to play a role as seed metrics in other signatures. As such Wick-rotations also can give neutral metrics, it is interesting to study the more general universal metrics in neutral signatures. In fact, one may wonder if all universal metrics are real sections of the general class of complex holomorphic universal metrics. 

Aside from the more mathematical classification issue, neutral metrics have also appeared in other geometric and  physics contexts, most notably, in twistor theory \cite{twist,twist2,twist3,twist4} and in $N=2$ string theory \cite{OoguriVafa}. Here we will focus on the Walker metrics as they appear as a separate branch in the study of degenerate metrics, for example, in the classification of four-dimensional neutral VSI spaces \cite{pseudoVSI2}. 

Our main result is given in Proposition \ref{prop:universal}, but we will first remind ourselves some of the notation and formalism used in four-dimensional neutral signature. The Walker metrics will then be introduced and our main results will then be presented in section \ref{sect:Walker}.

\subsection{Newman--Penrose formalism and boost weight decomposition in neutral signature}

The notation for spin coefficients for four-dimensional metrics of neutral signature
is adopted from \cite{Law2008}. In the neutral case, the null tetrad is a set
of four real null-vectors with the standard normalization convention
\begin{equation}
  \ell_a n^a = 1, \qquad m_a \tilde{m}^a = -1,
\end{equation}
while all other contractions vanish. Hence, the metric can be written in the form
\begin{equation}
  g_{ab} = 2 \ell_{(a} n_{b)} - 2 m_{(a} \tilde{m}_{b)}
  \label{eq:frame:metric}
\end{equation}
that is preserved under Lorentz transformations generated by null rotations
and the two independent boosts
\begin{equation}
  (\bl, \bn) \mapsto (e^{\lambda_1} \bl, e^{-\lambda_1} \bn), \qquad
  (\bm, \btm) \mapsto (e^{\lambda_2} \bm, e^{-\lambda_2} \btm).
  \label{eq:boosts}
\end{equation}
We shall say that a quantity $q$ has a boost weight ${\bf b} \equiv (b_1, b_2) \in {\mathbb Z}^2$
if it transforms under boosts \eref{eq:boosts} according to
\begin{equation}
  q \mapsto e^{b_1 \lambda_1 + b_2 \lambda_2} q.
\end{equation}
Thus, e.g., the frame vectors $\bl$, $\bn$, $\bm$ and $\btm$ have boost-weights
(1,0), ($-1$,0), (0,1) and (0,$-1$), respectively.

In the neutral case, the prime operation differs from the standard NP notation
and acts on the frame as
\begin{equation}
  (\bl)' = \bn, \quad (\bn)' = \bl, \quad (\bm)' = - \btm, \quad (\btm)' = - \bm.
\end{equation}
As a formal analogy to complex conjugation in the Lorentz case, one may consider
the application of tilde interchanging $\bm$ and $\btm$.

We decompose the covariant derivative in the frame as
\begin{equation}
  \D \equiv \ell^a \nabla_a, \quad
  \D' \equiv n^a \nabla_a, \quad
  \delta \equiv m^a \nabla_a, \quad
  \T \equiv \tilde{m}^a \nabla_a.
  \label{eq:frame:directderiv}
\end{equation}
The projection of the covariant derivative of the frame vectors onto the frame
leads to 24 independent spin coefficients denoted by $\kappa, \rho, \sigma, \tau, \alpha, \varepsilon$
and their prime, tilde and tilde-prime versions.
The spin coefficients are depicted in Figure \ref{fig:spincoeffs}
taking their corresponding boost-weights as coordinates in the boost-weight diagram. 
Note that the prime operation changes the boost-weight $(b_1, b_2)$ of a given quantity to $(-b_1, -b_2)$,
while the application of tilde to $(b_1, -b_2)$.

\begin{figure}
  \centering
  \caption{Boost-weight diagram of spin coefficients.}
  \label{fig:spincoeffs}

  \begin{tikzpicture}[scale=1]
    \tikzstyle{axes}=[]
    \tikzstyle{spincoeffs}=[only marks,mark=x,mark size=4pt,blue!90!black,thick]
    
    \begin{scope}[style=axes]
      \draw[->] (-3,0) -- (3,0) node[right] {$b_1$};
      \draw[->] (0,-3) -- (0,3) node[above] {$b_2$};
      \foreach \x/\xtext in {-2, -1, 1, 2}
        \draw[xshift=\x cm] (0pt,1pt) -- (0pt,-1pt); 
      \foreach \y/\ytext in {-2, -1, 1, 2}
        \draw[yshift=\y cm] (1pt,0pt) -- (-1pt,0pt); 
    \end{scope}
    
    \begin{scope}[style=spincoeffs]
      \draw plot coordinates{(2,1) (2,-1) (1,2) (1,0) (1,-2) (0,1) (0,-1) (-1,2) (-1,0) (-1,-2) (-2,1) (-2,-1)};
    \end{scope}

    \begin{scope}
      \draw (2.3,1) node {$\kappa$};
      \draw (2.3,-0.97) node {$\tilde\kappa$};
      \draw (1.3,2) node {$\sigma$};
      \draw (1,0.3) node {$\epsilon \ \ \rho$};
      \draw (1,-0.35) node {$\tilde\epsilon \ \ \tilde\rho$};
      \draw (1.3,-1.96) node {$\tilde\sigma$};
      \draw (0.05,1.3) node {$\tau \ \ \alpha'$};
      \draw (0.05,0.74) node {$\tilde\alpha \ \ \tilde\tau'$};
      \draw (0.05,-0.66) node {$\alpha \ \ \tau'$};
      \draw (0.05,-1.3) node {$\tilde\tau \ \ \tilde\alpha'$};
      \draw (-1.3,2.06) node {$\tilde\sigma'$};
      \draw (-0.96,0.3) node {$\epsilon' \ \ \rho'$};
      \draw (-1,-0.3) node {$\tilde\epsilon' \ \ \tilde\rho'$};
      \draw (-1.3,-1.94) node {$\sigma'$};
      \draw (-2.3,1.06) node {$\tilde\kappa'$};
      \draw (-2.3,-0.95) node {$\kappa'$};
    \end{scope}
  \end{tikzpicture}
\end{figure}
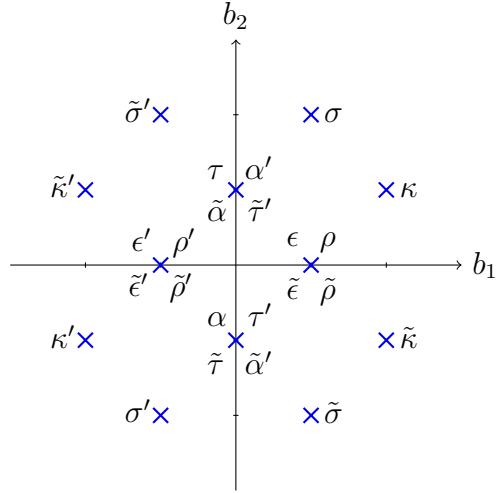

In four dimensions, the Ricci tensor has ten independent components which are encoded into
the following scalars
\begin{eqnarray}
  \Phi_{00} = \frac{1}{2} R_{ab} \ell^a \ell^b, \quad
  \Phi_{01} = \frac{1}{2} R_{ab} \ell^a m^b, \quad
  \Phi_{02} = \frac{1}{2} R_{ab} m^a m^b, \\
  \Phi_{10} = \frac{1}{2} R_{ab} \ell^a \tilde m^b, \quad
  \Phi_{11} = \frac{1}{4} R_{ab} \left( \ell^a n^b + m^a \tilde m^b \right), \quad
  \Phi_{12} = \frac{1}{2} R_{ab} n^a m^b, \\
  \Phi_{20} = \frac{1}{2} R_{ab} \tilde m^a \tilde m^b, \quad
  \Phi_{21} = \frac{1}{2} R_{ab} n^a \tilde m^b, \quad
  \Phi_{22} = \frac{1}{2} R_{ab} n^a n^b
\end{eqnarray}
and
\begin{equation}
  \Pi = \frac{R}{24}.
\end{equation}
Possible components of the Ricci tensor (and in fact of any rank-2 tensor) are only of boost-weights
$(0, 0)$, $(\pm 1, \pm 1)$, $(\pm 2, 0)$, $(0, \pm 2)$ and therefore they form a diamond-like shape
in the boost-weight diagram in Figure \ref{fig:Ricci}.

\begin{figure}
  \centering
  \caption{Boost-weight diagram of the Ricci tensor frame components.}
  \label{fig:Ricci}
  
  \begin{tikzpicture}[scale=1]
    \tikzstyle{axes}=[]
    \tikzstyle{ricci}=[only marks,mark=triangle,mark size=4pt,red!90!black,thick]
    
    \begin{scope}[style=axes]
      \draw[->] (-3,0) -- (3,0) node[right] {$b_1$};
      \draw[->] (0,-3) -- (0,3) node[above] {$b_2$};
      \foreach \x/\xtext in {-2, -1, 1, 2}
        \draw[xshift=\x cm] (0pt,1pt) -- (0pt,-1pt) node[below,fill=white] {$\xtext$};
      \foreach \y/\ytext in {-2, -1, 1, 2}
        \draw[yshift=\y cm] (1pt,0pt) -- (-1pt,0pt) node[left,fill=white] {$\ytext$};
    \end{scope}
    
    \begin{scope}[style=ricci]
      \draw plot coordinates{(2,0) (1,1) (1,-1) (0,2) (0,0) (0,-2) (-1,1) (-1,-1) (-2,0)};
    \end{scope}

    \begin{scope}
      \draw (2.4,0.3) node {$\Phi_{00}$};
      \draw (1.5,1) node {$\Phi_{01}$};
      \draw (1.5,-1) node {$\Phi_{10}$};
      \draw (0.5,2) node {$\Phi_{02}$};
      \draw (0.15,0.3) node {$\Pi$ \ $\Phi_{11}$};
      \draw (0.5,-2) node {$\Phi_{20}$};
      \draw (-1.5,1) node {$\Phi_{12}$};
      \draw (-1.5,-1) node {$\Phi_{21}$};
      \draw (-2.4,0.3) node {$\Phi_{22}$};
    \end{scope}
  \end{tikzpicture}
\end{figure}
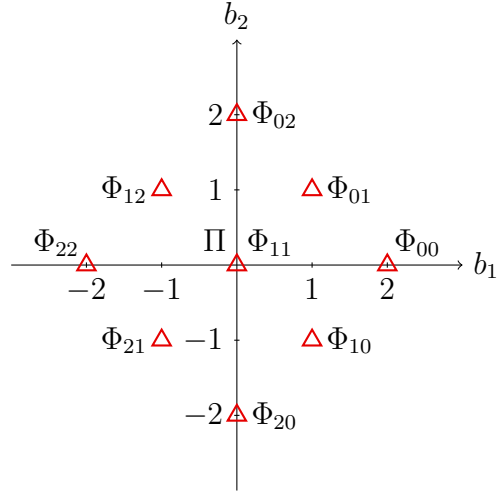

Ten independent components of the Weyl tensor are denoted by
\begin{eqnarray}
  \Psi_0 = C_{abcd} \ell^a m^b \ell^c m^d, \label{eq:Psi0} \\
  \Psi_1 = C_{abcd} \ell^a m^b \ell^c n^d, \\
  \Psi_2 = C_{abcd} \ell^a m^b \tilde{m}^c n^d, \label{eq:Psi2} \\
  \Psi_3 = C_{abcd} \ell^a n^b \tilde{m}^c n^d, \\
  \Psi_4 = C_{abcd} n^a \tilde{m}^b n^c \tilde{m}^d \label{eq:Psi4}
\end{eqnarray}
and their $\tilde\Psi_i$ counterparts obtained by applying tilde.
The Weyl tensor components form a cross-like shape in the boost-weight diagram in Figure \ref{fig:Weyl},
where the self-dual components $\Psi_i$ lie on the anti-diagonal,
while the anti-self-dual components $\tilde \Psi_i$ are placed on the diagonal.

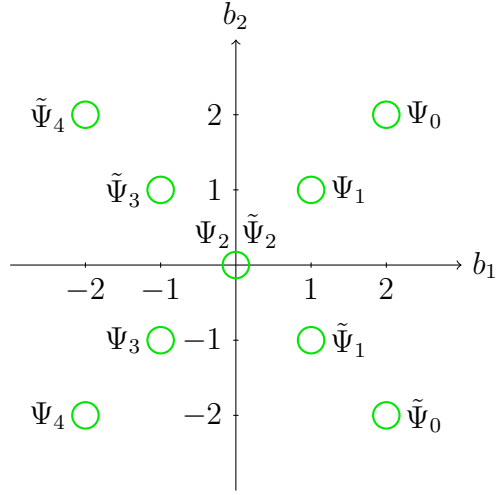
\begin{figure}
  \centering
  \caption{Boost-weight diagram of the Weyl tensor frame components.}
  \label{fig:Weyl}
  
  \begin{tikzpicture}[scale=1]
    \tikzstyle{axes}=[]
    \tikzstyle{weyl}=[only marks,mark=o,mark size=5pt,green!90!black,thick]
    
    \begin{scope}[style=axes]
      \draw[->] (-3,0) -- (3,0) node[right] {$b_1$};
      \draw[->] (0,-3) -- (0,3) node[above] {$b_2$};
      \foreach \x/\xtext in {-2, -1, 1, 2}
        \draw[xshift=\x cm] (0pt,1pt) -- (0pt,-1pt) node[below,fill=white] {$\xtext$};
      \foreach \y/\ytext in {-2, -1, 1, 2}
        \draw[yshift=\y cm] (1pt,0pt) -- (-1pt,0pt) node[left,fill=white] {$\ytext$};
    \end{scope}
    
    \begin{scope}[style=weyl]
      \draw plot coordinates{(2,2) (2,-2) (1,1) (1,-1) (0,0) (-1,1) (-1,-1) (-2,2) (-2,-2)};
    \end{scope}

    \begin{scope}
      \draw (2.5,2) node {$\Psi_0$};
      \draw (2.5,-2) node {$\tilde\Psi_0$};
      \draw (1.5,1) node {$\Psi_1$};
      \draw (1.5,-1) node {$\tilde\Psi_1$};
      \draw (0,0.45) node {$\Psi_2 \ \tilde\Psi_2$};
      \draw (-1.5,1) node {$\tilde\Psi_3$};
      \draw (-1.5,-1) node {$\Psi_3$};
      \draw (-2.5,2) node {$\tilde\Psi_4$};
      \draw (-2.5,-2) node {$\Psi_4$};
    \end{scope}
  \end{tikzpicture}
\end{figure}

Moreover, the NP formalism consists of three sets of equations.
The commutator relations and the Ricci identities can be found in \cite{Law2008}.
The full set of the Bianchi identities is given in \ref{appendix:bianchi}.

\section{Walker metrics}
\label{sect:Walker}
Metrics possessing invariant null $n$-planes with $n \le d / 2$
have been studied by Walker in \cite{Walker1950}.
We consider the case $d = 4$ with null planes of maximum dimensionality,
where in a real space the signature of metric is necessarily neutral.
The canonical form of four-dimensional Walker metrics admitting a field of parallel null 2-planes is \cite{Walker1950}
\begin{equation}
  \dd s^2 = 2 \dd u (\dd v + A \dd u + C \dd U) + 2 \dd U (\dd V + B \dd U),
  \label{eq:Walker_metric}
\end{equation}
where $A$, $B$ and $C$ are arbitrary functions.
Hence, the parallel null 2-planes are spanned by the null vectors $\partial_v$, $\partial_V$
and one can choose a natural real null frame
\begin{equation}
  \eqalign{
  \ell^a \partial_a = \partial_v, \quad
    n^a \partial_a = \partial_u - A \partial_v - \frac{1}{2} C \partial_V, \\
  \tilde{m}^a \partial_a = \partial_V, \quad
    m^a \partial_a = - \partial_U + B\partial_V + \frac{1}{2} C \partial_v.}
  \label{eq:Walker:frame}
\end{equation}
The spin coefficients can be immediately obtained by expressing
the directional derivatives of the frame vectors and comparing with
\eref{eq:Walker:derframevectors:first}--\eref{eq:Walker:derframevectors:last}.
It follows that most of the spin coefficients for the Walker metrics vanish
\begin{equation}
  \kappa = \rho = \tilde\kappa = \tilde\rho = \tilde\sigma = \tilde\tau = \sigma' = \tau' = \tilde\rho' = \tilde\tau' = \varepsilon = \tilde\varepsilon = \alpha = \tilde\alpha' = 0
  \label{eq:Walker:spincoeff:vanishing}
\end{equation}
and the ten remaining ones read
\begin{eqnarray}
  \sigma = -B_{,v}, \quad
  \tau = \frac{1}{2} C_{,v}, \quad
  \kappa' = - A_{,V}, \quad
  \rho' = - \frac{1}{2} C_{,V}, \label{eq:Walker:spincoeff:first} \\
  \tilde\kappa' = -B A_{,V} + A_{,U} - \frac{1}{2} C A_{,v} + \frac{1}{2} A C_{,v} - \frac{1}{2} C_{,u} + \frac{1}{4} C C_{,V}, \\
  \tilde\sigma' = A B_{,v} - B_{,u} + \frac{1}{2} C B_{,V} - \frac{1}{2} B C_{,V} + \frac{1}{2} C_{,U} - \frac{1}{4} C C_{,v}, \\
  \varepsilon' = \frac{1}{4} C_{,V} - \frac{1}{2} A_{,v}, \quad
  \tilde\varepsilon' = - \frac{1}{4} C_{,V} - \frac{1}{2} A_{,v}, \\
  \tilde\alpha = - \frac{1}{2} B_{,V} - \frac{1}{4} C_{,v}, \quad
  \alpha' = \frac{1}{2} B_{,V} - \frac{1}{4} C_{,v}. \label{eq:Walker:spincoeff:last}
\end{eqnarray}

Inspecting the Ricci identities \cite{Law2008} which are significantly simplified for the Walker metrics,
we are able to easily express the components of the trace-free Ricci tensor
\begin{eqnarray}
  \Phi_{00} = \Phi_{10} = \Phi_{20} = 0, \quad
  \Phi_{22} = - \T\tilde\kappa', \quad
  \Phi_{02} = \D\tilde\sigma', \quad
  \Phi_{01} = \D\tilde\alpha, \label{eq:Walker:Ricci:first} \\
  \Phi_{21} = -\T\tilde\varepsilon', \quad
  \Phi_{11} = \frac{\T\tilde\alpha - \D\tilde\varepsilon'}{2}, \quad
  \Phi_{12} = \frac{\T\tilde\sigma' - \D\tilde\kappa'}{2} \label{eq:Walker:Ricci:last}
\end{eqnarray}
and the Weyl tensor
\begin{eqnarray}
  \Psi_0 = - \D\sigma, \quad
  \Psi_1 = \D\alpha', \quad
  \Psi_2 = \frac{\T(\alpha' - \tau) - \D\varepsilon'}{3}, \label{eq:Walker:Weyl:first} \\
  \Psi_3 = -\T\varepsilon', \quad
  \Psi_4 = - \T\kappa', \\
  \tilde\Psi_0 = \tilde\Psi_1 = 0, \quad
  \tilde\Psi_2 = - \frac{\D\tilde\varepsilon' + \T\tilde\alpha}{3} = \frac{R}{12}, \quad
  \tilde\Psi_3 = - \frac{\D\tilde\kappa' + \T\tilde\sigma'}{2} \\
  \tilde\Psi_4 = 2(\tilde\sigma' \tilde\varepsilon' - \tilde\kappa' \tilde\alpha) - \delta\tilde\kappa' - \D'\tilde\sigma',
  \label{eq:Walker:Weyl:last}
\end{eqnarray}
Using \eref{eq:Walker:spincoeff:first}--\eref{eq:Walker:spincoeff:last},
these components can straightforwardly be rewritten in terms of the metric functions $A$, $B$, $C$
as given in \eref{eq:Walker:Phi00-11}--\eref{eq:Walker:tildePsi4}.

It is convenient to split the Weyl tensor $C = C^+ + C^-$ of a general Walker metric
to its self-dual part
\begin{equation}
  C^+ = C^{(-2,-2)} + C^{(-1,-1)} + C^{(0,0)^+} + C^{(1,1)} + C^{(2,2)},
  \label{eq:Walker:Weyl:SDpart}
\end{equation}
and anti-self-dual part
\begin{equation}
  C^- = C^{(-2,2)} + C^{(-1,1)} + C^{(0,0)^-}.
\end{equation}
The parts of the corresponding boost-weights are then given by
\begin{eqnarray}
  \fl
  C_{abcd}^{(2,2)} = 4 \Psi_0 \, n_{\{a} \tilde m_b n_c \tilde m_{d\}}, \label{eq:C(2,2)} \\
  \fl
  C_{abcd}^{(1,1)} = - 8 \Psi_1 \left( n_{\{a} \ell_b n_c \tilde m_{d\}} + n_{\{a} \tilde m_b m_c \tilde m_{d\}} \right), \\
  \fl
  C_{abcd}^{(0,0)^+} = 4 \Psi_2 \left( 2 \ell_{\{a} m_b \tilde m_c n_{d\}} - 2 \ell_{\{a} n_b m_c \tilde{m}_{d\}}
	      + \ell_{\{a} n_b \ell_c n_{d\}} + m_{\{a} \tilde{m}_b m_c \tilde{m}_{d\}} \right), \\
  \fl
  C_{abcd}^{(-1,-1)} = - 8 \Psi_3 \left( \ell_{\{a} n_b \ell_c m_{d\}} + \ell_{\{a} m_b \tilde m_c m_{d\}} \right), \\
  \fl
  C_{abcd}^{(-2,-2)} = 4 \Psi_4 \, \ell_{\{a} m_b \ell_c m_{d\}}, \label{eq:C(-2,-2)} \\
  \fl
  C_{abcd}^{(0,0)^-} = 4 \tilde\Psi_2 \left( 2 \ell_{\{a} \tilde m_b m_c n_{d\}} - 2 \ell_{\{a} n_b \tilde m_c m_{d\}}
    + \ell_{\{a} n_b \ell_c n_{d\}} + m_{\{a} \tilde{m}_b m_c \tilde{m}_{d\}} \right), \label{eq:C(0,0)-} \\
  \fl
  C_{abcd}^{(-1,1)} = - 8 \tilde\Psi_3 \left( \ell_{\{a} n_b \ell_c \tilde m_{d\}}
    + \ell_{\{a} \tilde m_b m_c \tilde m_{d\}} \right), \\
  \fl
  C_{abcd}^{(-2,2)} = 4 \tilde\Psi_4 \, \ell_{\{a} \tilde m_b \ell_c \tilde m_{d\}}, \label{eq:C(-2,2)}
\end{eqnarray}
where we introduce the operation \{\} defined as:
\begin{equation}
  u_{\{a} v_b w_c z_{d\}} \equiv \case{1}{2} (u_{[a} v_{b]} w_{[c} z_{d]} + u_{[c} v_{d]} w_{[a} z_{b]}).
\end{equation}
Note that $C^{(-2,2)}$, $C^{(-1,1)}$ and $C^{(0,0)^-}$ can be obtained from
$C^{(-2,-2)}$, $C^{(-1,-1)}$ and $C^{(0,0)^+}$, respectively, by using the tilde operation.

The non-vanishing spin coefficients and the frame components of the Ricci and Weyl tensors
of general Walker metrics \eref{eq:Walker_metric} are depicted in \fref{fig:Walker_metric}
according to their boost weights.
\begin{figure}
  \centering
  \caption{Boost weight diagram of a general Walker metric. The components of the Weyl tensor,
  Ricci tensor and spin coefficients are denoted by circles, triangles and crosses, respectively.}  
  \begin{tikzpicture}[scale=1]
    \tikzstyle{axes}=[]
    \tikzstyle{weyl}=[only marks,mark=o,mark size=5pt,green!90!black,thick]
    \tikzstyle{ricci}=[only marks,mark=triangle,mark size=4pt,red!90!black,thick]
    \tikzstyle{spincoeffs}=[only marks,mark=x,mark size=4pt,blue!90!black,thick]
    \begin{scope}[style=axes]
      \draw[->] (-3,0) -- (3,0) node[right] {$b_1$};
      \draw[->] (0,-3) -- (0,3) node[above] {$b_2$};
      \foreach \x/\xtext in {-2, -1, 1, 2}
        \draw[xshift=\x cm] (0pt,1pt) -- (0pt,-1pt) node[below,fill=white] {$\xtext$};
      \foreach \y/\ytext in {-2, -1, 1, 2}
        \draw[yshift=\y cm] (1pt,0pt) -- (-1pt,0pt) node[left,fill=white] {$\ytext$};
    \end{scope}
    \begin{scope}[style=weyl]
      \draw plot coordinates{(2,2) (1,1) (0,0) (-1,1) (-1,-1) (-2,2) (-2,-2)};
    \end{scope}
    \begin{scope}[style=ricci]
      \draw plot coordinates{(1,1) (0,2) (0,0) (-1,1) (-1,-1) (-2,0)};
    \end{scope}
    \begin{scope}[style=spincoeffs]
      \draw plot coordinates{(1,2) (0,1) (-1,2) (-1,0) (-2,1) (-2,-1)};
    \end{scope}
  \end{tikzpicture}
  \label{fig:Walker_metric}
\end{figure}
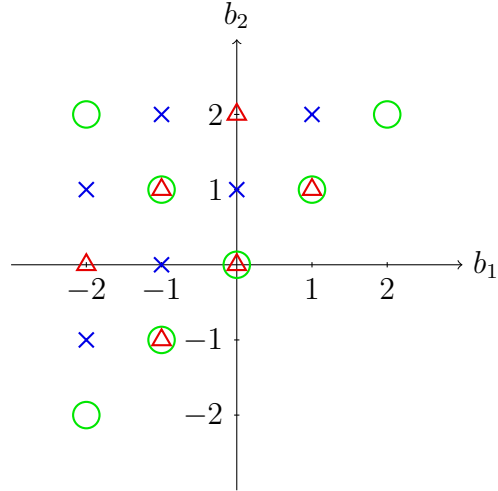

\subsection{Einstein CSI Walker metrics}
\label{section:EinsteinCSI}

Since universal metrics are necessarily Einstein, i.e.\ the trace-free Ricci tensor vanishes
and the Ricci scalar is constant, from now on we assume
\begin{equation}
  R_{ab} = \lambda g_{ab}.
  \label{eq:Einstein_space}
\end{equation}

Let us also elaborate on the assumption put on the components $\Psi_i$ and $\tilde{\Psi}_2$.
As mentioned in the introduction, universal spaces are necessarily CSI. The Weyl components
of CSI spaces can be of the following types. There exists a frame so that: 
\begin{enumerate}
\item All $\Psi_i$ and $\tilde{\Psi}_i$ are constants. 
\item $\Psi_i$ are constants, $\tilde{\Psi}_2$ is constant, and $\tilde{\Psi}_1=\tilde{\Psi}_0=0$. 
\item $\tilde{\Psi}_i$ are constants, ${\Psi}_2$ is constant, and ${\Psi}_1={\Psi}_0=0$. 
\item $\Psi_2$ and $\tilde{\Psi}_2$ are constant, and ${\Psi}_1={\Psi}_0=\tilde{\Psi}_1=\tilde{\Psi}_0=0$. 
\end{enumerate}
The Walker metrics we consider here, are therefore in category 2. The first category is reminiscent (but not exclusive) of locally homogeneous space while Kundt spaces are in category 4. 

These assumptions on the components of the Ricci and Weyl tensors allow us to fully determine
the $v$- and $V$-dependence of the metric functions $A$, $B$ and $C$ by integrating
\eref{eq:Walker:Phi00-11} and \eref{eq:Walker:Phi01,21}
along with \eref{eq:Walker:Psi0-2}--\eref{eq:Walker:tildePsi0-2}. Thus, we obtain
\begin{eqnarray}
  A = \frac{1}{2} \Psi_4 V^2 + \Psi_3 v V + \left( \frac{\lambda}{3} + \frac{\Psi_2}{2} \right) v^2 + A_{10} v + A_{01} V + A_{00}, \\
  B = \frac{1}{2} \Psi_0 v^2 + \Psi_1 v V + \left( \frac{\lambda}{3} + \frac{\Psi_2}{2} \right) V^2 + B_{10} v + B_{01} V + B_{00}, \\
  C = - \Psi_3 V^2 - \Psi_1 v^2 + 2 \left( \frac{\lambda}{3} - \Psi_2 \right) v V + C_{10} v + C_{01} V + C_{00},
\end{eqnarray}
where $A_{10}$, $A_{01}$, $A_{00}$, $B_{10}$, $B_{01}$, $B_{00}$, $C_{10}$, $C_{01}$, and $C_{00}$
are arbitrary functions of $u$, $U$.

However, the conditions that universal metrics are Einstein and CSI are only necessary, 
but not sufficient, and we still have to analyze the components of the Weyl tensor
and check all their possible contributions to the conserved symmetric rank-2 tensors.

For convenience, one may introduce the following definition:

\begin{definition}
  A boost weight difference $\bdiff(q)$ of a quantity $q$ with boost weights $b = (b_1, b_2)$
  is the difference $\bdiff(q) = b_2 - b_1$.
\end{definition}

Boost weights are additive and hence the boost weight difference of the product of scalars
$\eta \mu$ is the sum $\bdiff(\eta) + \bdiff(\mu)$.

Let us now mention some useful implications of the boost-weight decomposition for rank-2 tensors.
Components of a general rank-2 tensor are only of b.w.\
$(0,0)$, $(\pm 1, \pm1)$, $(\pm 2,0)$ and $(0, \pm 2)$.
Therefore, these components are only of boost-weight difference $-2$, 0 or 2 and   
form a diamond-like shape in the boost-weight diagram, as in the case of Ricci tensor in \fref{fig:Ricci}.
Now, consider rank-2 tensors polynomial in a tensor $T$.
Due to the additivity of boost-weights,
if $T$ has only components with $\bdiff > 2$, any tensor polynomial in $T$ has only components with $\bdiff > 2$
and thus $T$ does not contribute to the rank-2 tensors.
Similarly, if the components of $T$ are of $\bdiff > 1$ or $\bdiff > 0$,
the rank-2 tensors are at most linear or quadratic in $T$, respectively.

As we already pointed out, if a tensor $T$ has, for example, only components with $\bdiff = 3$,
it does not contribute to the rank-2 tensors. However, the covariant derivative of $T$ in principle
could generate components of $\bdiff = 2$ which contribute. Specifically, for a given component $\eta$
of boost weights $(b_1, b_2)$, the application of the covariant derivative yields components
$\D\eta$, $\T\eta$, $\D'\eta$ and $\delta\eta$ of boost weights $(b_1+1,b_2)$, $(b_1,b_2-1)$, $(b_1-1,b_2)$
and $(b_1,b_2+1)$, respectively. Obviously, $\bdiff(\D\eta) = \bdiff(\T\eta) = \bdiff(\eta) - 1$
and $\bdiff(\D'\eta) = \bdiff(\delta\eta) = \bdiff(\eta) + 1$.
Therefore the covariant derivative of a tensor $T$ of arbitrary order,
could give rise in general to components of arbitrary $\bdiff$ and thus to the infinite number
of possible contributions to the rank-2 tensors. 
On the other hand, it turns out that in special cases, the values of $\bdiff$ of these components are restricted.
Before we show this using balanced scalars, it is convenient to introduce the following notation:

\begin{definition}
  We use the symbol $\calD$ as a derivative operator representing
  both directional derivatives $\D$ and $\T$,
  i.e.\ $\calD \eta = 0$ means that $\D\eta = 0$ and simultaneously $\T\eta = 0$.
  Similarly, $\calD'$ represents both $\D'$ and $\delta$.
\end{definition}

Note that $\D$ and $\T$ commutes 
and thus, for example,
$\calD^k \eta = 0$ corresponds to $k + 1$ equations $\D^p\T^q \eta = 0$, where $p + q = k$.

Inspired by the balanced scalar approach employed in the Lorentzian case in
\cite{PravdaPravdovaColeyMilson2002}, we introduce the notion of dual-balance
reflecting the symmetry of neutral signature Walker metrics \eref{eq:Walker_metric}.

\begin{definition}
  We say that a scalar $\eta$ is $k$-dual-balanced if
  \begin{eqnarray}
    \eta = 0 \quad\textrm{for}\quad \bdiff(\eta) \le k, \label{def:balanced:cond1} \\
    \mathcal{D}^{\bdiff(\eta) - k} \eta = 0 \quad\textrm{for}\quad \bdiff(\eta) > k. \label{def:balanced:cond2}
  \end{eqnarray}
  A tensor is $k$-dual-balanced if all its components are $k$-dual-balanced scalars.
\end{definition}

Now, let us inspect the Weyl tensor. As a direct consequence of Lemma \ref{lemma:covd} from \ref{appendix:balanced},
if a tensor is 2-dual-balanced, the covariant derivative of any order of this tensor 
is 2-dual-balanced as well. In other words, for a given 2-dual-balanced tensor,
all components of this tensor and its covariant derivative of arbitrary order
are of boost-weight difference $\bdiff > 2$ and therefore they do not contribute
to any rank-2 tensor. However, not all of the components \eref{eq:C(2,2)}--\eref{eq:C(-2,2)}
of the Weyl tensor of Einstein Walker metrics with a constant self-dual part
are 2-dual-balanced.

For Einstein Walker metrics, $C^{(-2,2)}$ is 2-dual-balanced since $\bdiff(\tilde\Psi_4) = 4$
and \eref{eq:Walker:EinsteinASD:DtildePsi4} holds. From \eref{eq:Walker:EinsteinASD:DtildePsi3}
along with $\bdiff(\tilde\Psi_3) = 2$, it follows that $C^{(-1,1)}$ is 1-dual-balanced.
Obviously, $C^{(0,0)^-}$ and the self-dual part \eref{eq:Walker:Weyl:SDpart}
of the Weyl tensor with constant $\Psi_i$ are $-1$-dual-balanced. On the other hand,
\begin{equation}
  C_{abcd;e}^{(0,0)^-} = 24 \tilde\Psi_2 (\tilde\sigma' \tilde m_e - \tilde\kappa' \ell_e) \left( \ell_{\{a} n_b \ell_c \tilde m_{d\}} + \ell_{\{a} \tilde m_b m_c \tilde m_{d\}} \right)
  \label{eq:Einstein:nablaC(0,0)-}
\end{equation}
has only components of $\bdiff = 3$. Using
\eref{eq:Walker:EinsteinASD:derspincoeff:first}--\eref{eq:Walker:EinsteinASD:DtildePsi3}, it holds that
$\calD^2(\tilde\kappa' \tilde\Psi_2) = \calD^2(\tilde\sigma' \tilde\Psi_2) = 0$
implying $\nabla C^{(0,0)^-}$ is 1-dual-balanced. Therefore, the Weyl tensor
of Einstein Walker metrics with a constant self-dual part is in general
$-1$-dual-balanced and thus $\nabla^{(k)}C$ has no components of negative boost-weight difference.

Since $C^{(-2,2)}$ is 2-dual-balanced, $C^{(-2,2)}$ and its derivatives contain
components of $\bdiff > 2$ and thus they do not contribute to rank-2 tensors
constructed from the Weyl tensor and its derivatives of arbitrary order.
Similarly, $\nabla^{(k)}C^{(-1,1)}$ with $k \ge 0$ and $\nabla^{(k)}C^{(0,0)^-}$
with $k > 0$ contain components of $\bdiff > 1$ implying that the contributions
to such rank-2 tensors are at most linear in $\nabla^{(k)}C^{(-1,1)}$ with
$k \ge 0$ or $\nabla^{(k)}C^{(0,0)^-}$ with $k > 0$. However, any number of
$C^{(0,0)} = C^{(0,0)^+} + C^{(0,0)^-}$ and many combinations of $\nabla^{(k)}C^+$
can contribute to such rank-2 tensors. All the potential non-vanishing contributions
are treated in detail in the following sections.

\subsection{0-universal Walker metrics}
\label{section:0universalWalker}

Before we proceed to identify classes of universal Einstein Walker metrics with 
a constant self-dual part of the Weyl tensor, we focus on symmetric rank-2 tensors
constructed from the metric and the Weyl tensor {\em without} derivatives,
i.e.\ on 0-universal metrics.

First of all, let us consider a discrete frame transformation
\begin{equation}
  (\bn, \bl) \leftrightarrow (-\btm, \bm).
  \label{eq:pair_interchange2}
\end{equation}
Under the interchange of the pairs of frame vectors \eref{eq:pair_interchange2},
the frame normalization and thus the form of the metric \eref{eq:frame:metric}
is preserved. However, the frame components and the corresponding parts of the Weyl
tensor may be altered by \eref{eq:pair_interchange2} since they depend on the choice
of the frame. As can be seen directly from \eref{eq:Psi0}--\eref{eq:Psi4},
$\Psi_0$, $\Psi_2$, $\tilde\Psi_2$, $\Psi_4$ are invariant and
\begin{equation}
  \Psi_1 \mapsto -\Psi_1, \quad \Psi_3 \mapsto -\Psi_3.
\end{equation}
Subsequently, all $C^{(b,b)}$ with $b \in \{-2,\dots,2\}$ given by \eref{eq:C(2,2)}--\eref{eq:C(0,0)-}
remain unchanged.

The pairs of frame vectors may also be interchange as
\begin{equation}
  (\bn, \bl) \leftrightarrow (-\bm, \btm),
  \label{eq:pair_interchange}
\end{equation}
then $\Psi_2$, $\tilde\Psi_2$ are invariant and
\begin{equation}
  \Psi_0 \leftrightarrow \Psi_4, \quad \Psi_1 \leftrightarrow \Psi_3,
  \quad \tilde\Psi_3 \mapsto - \tilde\Psi_3.
\end{equation}

Several useful implications of \eref{eq:pair_interchange2} and \eref{eq:pair_interchange}
for rank-2 tensors can be found in \ref{appendix:tranfs}. These lemmas allow us to rule out 
potentially non-vanishing components of symmetric rank-2 tensors constructed from the metric
and the Weyl tensor without derivatives.

As already mentioned above, since the Weyl
tensor is $-1$-dual-balanced, it has no components of negative boost-weight
difference and therefore only components of $\bdiff = 0$ and $\bdiff = 2$ may
appear in such rank-2 tensors. It remains to show that b.w.\ (0,0) parts of the
rank-2 tensors are proportional to the metric and b.w.\ (1,1), $(-1,-1)$, $(-1,1)$,
(0,2) and $(-2,0)$ parts vanish.

The components of $\bdiff = 0$ only consist of $\Psi_i$ and $\tilde\Psi_2$,
i.e.\ they arise from the self-dual part \eref{eq:Walker:Weyl:SDpart} and $C^{(0,0)^-}$.
Hence, the cases to consider are:
\begin{itemize}
  \item b.w.\ (0,0):
    The rank-2 tensors take one of the forms
    \begin{equation*}
      \fl
      C_{\circ\circ\circ\circ}^{(b_1,b_1)} \cdots C_{\circ\circ\circ\circ}^{(b_n,b_n)} g_{ab}, \quad
      C_{a\circ b\circ}^{(b_1,b_1)} \cdots C_{\circ\circ\circ\circ}^{(b_n,b_n)}, \quad
      C_{a\circ\circ\circ}^{(b_1,b_1)} C_{b \circ\circ\circ}^{(b_2,b_2)} \cdots C_{\circ\circ\circ\circ}^{(b_n,b_n)},
    \end{equation*}
    where all indices ``$\circ$'' are contracted.
    Since all $\Psi_i$ and $\tilde\Psi_2$ are constant, such tensors are proportional to the metric with a constant factor
    as it is obvious for the first form and follows from Lemma
    \ref{lemma:C(0,0)^k} for the remaining two forms.
  \item b.w.\ (1,1) and ($-1$,$-1$):
    These components consist of only $\Psi_1$ or $\Psi_3$, respectively,
    and any combination of b.w.\ (0,0) blocks \eref{eq:bw(0,0)_blocks},
    therefore, they change the sign under \eref{eq:pair_interchange2}.
    Then, Lemma \ref{lemma:T(1,1)} implies that b.w.\ (1,1) and ($-1$,$-1$) parts
    of a symmetric rank-2 tensor vanish.
\end{itemize}
The components of $\bdiff = 2$ arise from one $C^{(-1,1)}$ and certain combinations of 
self-dual part \eref{eq:Walker:Weyl:SDpart} and $C^{(0,0)^-}$.
\begin{itemize}
  \item b.w.\ ($-1$,1):
    The components of b.w.\ ($-1$,1) are made only of $\tilde\Psi_3$ and any combination
    of b.w.\ (0,0) blocks \eref{eq:bw(0,0)_blocks} which are invariant under
    \eref{eq:pair_interchange} except $\Psi_1^2\Psi_4 \leftrightarrow \Psi_0\Psi_3^2$.
    Hence, we require $\Psi_1^2\Psi_4 = \Psi_0\Psi_3^2$ to employ Lemma \ref{lemma:T(-1,1)} implying
    that the b.w.\ ($-1$,1) part of symmetric rank-2 tensors vanishes.
  \item b.w.\ (0,2) and ($-2$,0):
    These parts of a rank-2 tensor $T_{ab}$ are symmetric and their components involve only
    $\Psi_1 \tilde\Psi_3$ or $\Psi_3 \tilde\Psi_3$, respectively, and any combination
    of b.w.\ (0,0) blocks \eref{eq:bw(0,0)_blocks}. These b.w.\ (0,0) blocks are, as a result of the proof of Lemma \ref{lemma:C(0,0)^k}, invariant under \eref{eq:pair_interchange}, as well as the boosts. Hence, there is a transitive action group on the tangent space leaving the b.w.\ (0,0) blocks invariant. This implies that the b.w.\ (0,0) blocks  $T_{22} = T_{ab} \, m^a m^b$ and
    and $T_{11} = T_{ab} \, n_a n_b$ are equal. 

    Since $T_{11}$ and $T_{22}$ have the same b.w.\ (0,0) blocks that are invariant under
    \eref{eq:pair_interchange} if we assume $\Psi_1^2\Psi_4 = \Psi_0\Psi_3^2$,
    $\tilde\Psi_3$ changes the sign and $\Psi_1 \leftrightarrow \Psi_3$,
    it holds that $T_{11} \leftrightarrow - T_{22}$.
    Then, it immediately follows from Lemma \ref{lemma:T(0,2)} that these components vanish.
\end{itemize}

Therefore, we arrive at:
\begin{proposition}
  Einstein Walker metrics with a constant self-dual part of the Weyl tensor are
  0-universal if $\tilde\Psi_3 = 0$ or $\Psi_1^2 \Psi_4 = \Psi_0 \Psi_3^2$.
\end{proposition}

\subsection{Universal Walker metrics}

Let us start this section by introducing two definitions that simplify our notation.
\begin{definition}
  The symbol $\xi$ represents any $\bdiff = 1$ spin coefficient $\tau$, $\tilde\alpha$,
  $\alpha'$, $\rho'$, $\varepsilon'$, $\tilde\varepsilon'$, $\sigma$ and $\kappa'$.
\end{definition}
\begin{definition}
  If we enclose any number of directional derivatives in braces, such an object
  represents any permutation of the directional derivatives inside the braces.
\end{definition}
For instance, $\{\calD \calD'\}$ corresponds to any of $\D\D'$, $\D'\D$, $\D\delta$,
$\delta\D$, $\T\D'$, $\D'\T$, $\T\delta$ and $\delta\T$.

Now, we are able to list efficiently all components of $\nabla^{(k)}C$ contributing to rank-2
tensors. No $\bdiff < 0$ components arise from the Weyl tensor and its derivatives
since it is at least $-1$-dual-balanced, therefore, we only focus on the relevant
components of $0 \le \bdiff \le 2$. Employing Lemmas \ref{lemma:DD'mu} and
\ref{lemma:DD'xi} from \ref{appendix:DC}, the potential contributions from derivatives of the self-dual
part of the Weyl tensor $\nabla^{(n)} C^+$ are:
\begin{itemize}
  \item $\bdiff = 0$: $\Psi_i (\calD\xi)^k$ with $n = 2k$ formed by $\Psi_i$ and
    any number of $\bdiff = 0$ terms $\{\calD^l \calD'^{l-1}\} \xi$.
  \item $\bdiff = 1$: $\Psi_i (\calD\xi)^k \xi$ with $n = 2k + 1$ given by $\Psi_i$,
    $\bdiff = 1$ term $\{\calD^l \calD'^l\}\xi$ and any number of $\bdiff = 0$ terms
    $\{\calD^m \calD'^{m-1}\} \xi$.
  \item $\bdiff = 2$: These components may be constructed in two ways. Either from $\Psi_i$,
    two $\bdiff = 1$ terms $\{\calD^l \calD'^l\}\xi$ and any number of $\bdiff = 0$ terms
    $\{\calD^m \calD'^{m-1}\} \xi$ leading to $\Psi_i (\calD\xi)^{k-2} \xi^2$ with $n = 2k$.
    They also may be made of $\Psi_i$, one of the $\bdiff = 2$ terms $\{\calD^l \calD'^{l-1}\}\tilde\sigma'$
    or $\{\calD^l \calD'^{l-1}\}\tilde\kappa'$ and any number of $\bdiff = 0$ terms $\{\calD^m \calD'^{m-1}\} \xi$
    yielding $\Psi_i \tilde\Psi_3 (\calD\xi)^k$ with $n = 2k$.
\end{itemize}

The relevant components of derivatives of the anti-self-dual part of the Weyl tensor
$\nabla^{(2n)}C_{abcd}^{(0,0)^-}$ and $\nabla^{(2n)}C_{abcd}^{(-1,1)}$
contributing to rank-2 tensors are only of $\bdiff = 2$ and due to Lemma \ref{lemma:DC}
from \ref{appendix:DC}
take the forms $\tilde\Psi_2 \tilde\Psi_3 \sum(\calD\xi)^{n-1}$ and $\tilde\Psi_3 \sum(\calD\xi)^n$,
respectively.

Therefore, symmetric rank-2 tensors constructed from the metric, the Weyl tensor
and its derivatives of any order consist of the following components:
\begin{itemize}
  \item $\bdiff = 0$:
    In addition to the cases with no derivatives of the Weyl tensor already mentioned
    in section \ref{section:0universalWalker}, components coming from derivatives
    of the self-dual part of the Weyl tensor appear. These terms take the form
    $\Psi_i (\calD\xi)^k$ and due to
    \eref{eq:Walker:EinsteinASD:derspincoeff:first}--\eref{eq:Walker:EinsteinASD:derspincoeff:last}
    they yield effectively the same contributions as the cases of b.w.\ (0,0),
    (1,1) and ($-1$,$-1$) in section \ref{section:0universalWalker}.
    Therefore, the $\bdiff = 0$ part is proportional to the metric.
  \item $\bdiff = 2$:
    The components of $\bdiff = 2$ can be constructed from any number of $\bdiff = 0$
    terms and either two $\bdiff = 1$ terms as
    \begin{equation}
      \Psi_i \Psi_j (\calD\xi)^k \xi^2,
      \label{eq:rank2:bdiff2}
    \end{equation}
    or one $\bdiff = 2$ term as
    \begin{equation}
      \Psi_i (\calD\xi)^m \xi^2, \quad \Psi_i \tilde\Psi_3 (\calD\xi)^m
      \label{eq:rank2:bdiff2a}
    \end{equation}
    and
    \begin{equation}
      \tilde\Psi_3 (\calD\xi)^m, \quad \tilde\Psi_2 \tilde\Psi_3 (\calD\xi)^m,
      \label{eq:rank2:bdiff2b}
    \end{equation}
    where the $\bdiff = 2$ terms in \eref{eq:rank2:bdiff2a} and \eref{eq:rank2:bdiff2b}
    come from derivatives of the self-dual part and the anti-self-dual part,
    respectively. The metric is universal if all these $\bdiff = 2$ components
    of rank-2 tensors vanish. However, there are many non-vanishing terms arising
    from \eref{eq:rank2:bdiff2} and \eref{eq:rank2:bdiff2a} for which Lemmas
    \ref{lemma:T(1,1)}--\ref{lemma:T(0,2)} from \ref{appendix:tranfs} are not applicable.
    For example, the b.w.\ ($-1$,1) term $\Psi_0 \Psi_4 \varepsilon' \tau$ transforms
    under \eref{eq:pair_interchange} to $\Psi_0 \Psi_4 \alpha' \rho'$, etc.
    For this reason, let us assume $\Psi_i = 0$. Then the only potentially
    contributing $\bdiff = 2$ components \eref{eq:rank2:bdiff2b} are of b.w.\ $(-1,1)$, 
    take the form $P(\lambda) \tilde\Psi_3$ and vanish due to Lemma \ref{lemma:T(-1,1)}.
\end{itemize}

Finally, we can conclude that, for Einstein Walker metrics with a vanishing self-dual part of the Weyl tensor,
all symmetric rank-2 tensors constructed from the metric, the Weyl tensor and its derivatives
of any order are proportional to the metric and hence:

\begin{proposition}
  Einstein Walker metrics with a vanishing self-dual part of the Weyl tensor are universal.
  \label{prop:universal}
\end{proposition}

Boost-weight diagram of universal Walker metrics is depicted in Figure \ref{fig:universal_Walker_metric}.
Note that the condition for the vanishing of self-dual part of the Weyl tensor is related
to our identification of the vectors $\partial_v$, $\partial_V$ generating the parallel null 2-planes
with the frame vectors $\bl$, $\btm$.
If we choose, instead, e.g.\ $\bl$ and $\bm$, the analogous statement holds for the vanishing of the anti-self-dual part.

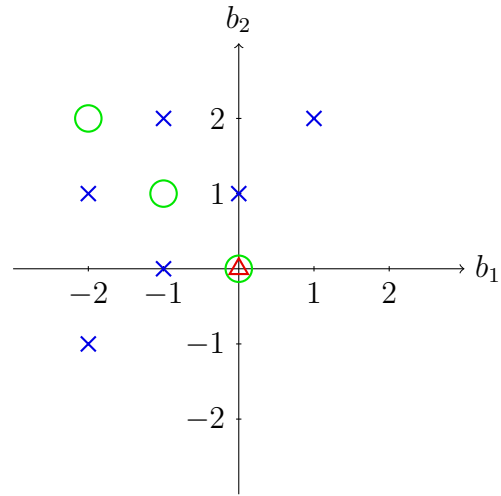
\begin{figure}
  \centering
  \caption{Boost weight diagram of Einstein Walker metrics with anti-self-dual Weyl tensor.
    The components of the Weyl tensor, Ricci tensor and spin coefficients are denoted by circles,
    triangles and crosses, respectively.}
  \begin{tikzpicture}[scale=1]
    \tikzstyle{axes}=[]
    \tikzstyle{weyl}=[only marks,mark=o,mark size=5pt,green!90!black,thick]
    \tikzstyle{ricci}=[only marks,mark=triangle,mark size=4pt,red!90!black,thick]
    \tikzstyle{spincoeffs}=[only marks,mark=x,mark size=4pt,blue!90!black,thick]
    \begin{scope}[style=axes]
      \draw[->] (-3,0) -- (3,0) node[right] {$b_1$};
      \draw[->] (0,-3) -- (0,3) node[above] {$b_2$};
      \foreach \x/\xtext in {-2, -1, 1, 2}
        \draw[xshift=\x cm] (0pt,1pt) -- (0pt,-1pt) node[below,fill=white] {$\xtext$};
      \foreach \y/\ytext in {-2, -1, 1, 2}
        \draw[yshift=\y cm] (1pt,0pt) -- (-1pt,0pt) node[left,fill=white] {$\ytext$};
    \end{scope}
    \begin{scope}[style=weyl]
      \draw plot coordinates{(0,0) (-1,1) (-2,2)};
    \end{scope}
    \begin{scope}[style=ricci]
      \draw plot coordinates{(0,0)};
    \end{scope}
    \begin{scope}[style=spincoeffs]
      \draw plot coordinates{(1,2) (0,1) (-1,2) (-1,0) (-2,1) (-2,-1)};
    \end{scope}
  \end{tikzpicture}
  \label{fig:universal_Walker_metric}
\end{figure}

Employing Proposition \ref{prop:universal}, it follows that the Walker metrics \eref{eq:Walker_metric}
are universal if the functions $A$, $B$, $C$ take the form
\begin{eqnarray}
\label{metricA}
  A = \frac{\lambda}{3} v^2 + A_{10} v + A_{01} V + A_{00}, \\
\label{metricB}
  B = \frac{\lambda}{3} V^2 + B_{10} v + B_{01} V + B_{00}, \\
\label{metricC}
  C = 2 \frac{\lambda}{3} v V + C_{10} v + C_{01} V + C_{00},
\end{eqnarray}
where $A_{10}$, $A_{01}$, $A_{00}$, $B_{10}$, $B_{01}$, $B_{00}$, $C_{10}$, $C_{01}$, and $C_{00}$
are arbitrary functions of $u$, $U$.

Although the null geodetic vectors $\bl = \partial_v$ and $\btm = \partial_V$
are non-expanding and non-twisting,
the projected trace-free symmetric parts of their covariant derivatives
\begin{equation}
  \sigma^{(\bl)}_{ab} = - \sigma \, \tilde m_a \tilde m_b, \quad
  \sigma^{(\btm)}_{ab} = - \kappa' \, \ell_a \ell_b,
\end{equation}
corresponding to the shear are, in general, non-vanishing.
Therefore, Walker metrics do not belong to the Kundt class unless $\sigma$ or $\kappa'$ vanishes,
i.e., using \eref{eq:Walker:spincoeff:first}, $B$ is independent of $v$ ($B_{10} = 0$)
or $A$ is independent of $V$ ($A_{01} = 0$).
The same observation can be also made directly by comparing the Walker metric \eref{eq:Walker_metric}
with the canonical form of the Kundt metrics \eref{eq:Kundt}.

\section{Conclusion}

We have considered four-dimensional Walker metrics of neutral signature and identified a wide subclass
of these metrics being universal, i.e., solving the vacuum field equations of any higher-order gravity theory.
Before this study, all known examples of universal metrics were locally homogeneous or Kundt.
It turns out that universal Walker metrics need not be Kundt in general, nor need they be locally homogeneous, and thus
represent a new class of universal spaces which has no analogue in the Riemannian and Lorentzian signatures. Explicitly, these universal metrics are Walker metrics \eref{eq:Walker_metric} with metric functions of the form \eref{metricA}--\eref{metricC}.

\section*{Acknowledgements} 
SH was supported through the Research Council of Norway, Toppforsk
grant no. 250367: \emph{Pseudo-Riemannian Geometry and Polynomial Curvature Invariants:
Classification, Characterisation and Applications.}
TM acknowledges support from the Albert Einstein Center for Gravitation and Astrophysics,
Czech Science Foundation GACR 14-37086G.
  
\appendix

\section{Bianchi identities}
\label{appendix:bianchi}

Contractions of the Bianchi identity $R_{ab[cd;e]} = 0$ with appropriate combinations of the frame vectors give rise to 
\begin{eqnarray}
  \fl
  \T\Psi_0 - \D\Psi_1
  - \delta\Phi_{00} + \D\Phi_{01}
  = 
  (4 \alpha + \tau') \Psi_0
  - 2 (\varepsilon - 2 \rho) \Psi_1
  - 3 \kappa \Psi_2 \nonumber \\
  - (2 \alpha' + 2 \tilde\alpha + \tilde\tau') \Phi_{00}
  + 2 (\varepsilon - \tilde\rho) \Phi_{01}
  - 2 \sigma \Phi_{10}
  + \tilde\kappa \Phi_{02}
  + 2 \kappa \Phi_{11},
  \label{eq:Bianchi:1} \\
  \fl
  \T\Psi_1 - \D\Psi_2
  - \D'\Phi_{00} + \T\Phi_{01}
  - 2 \D\Pi
  =
  \sigma' \Psi_0
  + 2 (\alpha + \tau') \Psi_1
  + 3 \rho \Psi_2
  - 2 \kappa \Psi_3 \nonumber \\
  + (2 \varepsilon' + 2 \tilde\varepsilon' + \tilde\rho') \Phi_{00}
  + 2 (\alpha - \tilde{\tau}) \Phi_{01}
  + \tilde\sigma \Phi_{02}
  - 2 \tau \Phi_{10}
  + 2 \rho \Phi_{11},
  \label{eq:Bianchi:2} \\
  \fl
  \T\Psi_2 - \D\Psi_3
  - \delta\Phi_{20} + \D\Phi_{21}
  + 2 \T\Pi
  =
  2 \sigma' \Psi_1
  + 3 \tau' \Psi_2
  + 2 (\varepsilon + \rho) \Psi_3
  - \kappa \Psi_4 \nonumber \\
  - 2 \rho' \Phi_{10}
  - 2 \tau' \Phi_{11}
  + (2\alpha' - 2 \tilde\alpha - \tilde\tau') \Phi_{20}
  - 2 (\varepsilon + \tilde\rho) \Phi_{21}
  + \tilde\kappa \Phi_{22},
  \label{eq:Bianchi:3} \\
  \fl
  \T\Psi_3 - \D\Psi_4
  - \D'\Phi_{20} + \T\Phi_{21}
  =
  3 \sigma' \Psi_2
  - 2 (\alpha - 2 \tau') \Psi_3
  + (4 \varepsilon + \rho) \Psi_4 \nonumber \\
  + 2 \kappa' \Phi_{10}
  + 2 \sigma' \Phi_{11}
  - (2 \varepsilon ' - 2 \tilde\varepsilon' - \tilde\rho') \Phi_{20}
  - 2 (\alpha + \tilde\tau) \Phi_{21}
  + \tilde\sigma \Phi_{22}.
  \label{eq:Bianchi:4}
\end{eqnarray}
The full set of equations can be obtained from \eref{eq:Bianchi:1}--\eref{eq:Bianchi:4}
using the prime operation and the application of tilde, where one must take into account
the effects of these manipulations on the frame components of the curvature tensors.
Under the prime operation, $\Psi_2$ and $\Phi_{11}$ are unchanged while
\begin{eqnarray}
  \Psi_0 \leftrightarrow \Psi_4, \quad
  \Psi_1 \leftrightarrow - \Psi_3, \\
  \Phi_{00} \leftrightarrow \Phi_{22}, \quad
  \Phi_{01} \leftrightarrow - \Phi_{21}, \quad
  \Phi_{02} \leftrightarrow \Phi_{20}, \quad
  \Phi_{10} \leftrightarrow - \Phi_{12}. \quad
\end{eqnarray}
Applying tilde, $\Phi_{00}$, $\Phi_{11}$, $\Phi_{22}$ remain unaltered and
\begin{equation}
  \Phi_{01} \leftrightarrow \Phi_{10}, \quad
  \Phi_{02} \leftrightarrow \Phi_{20}, \quad
  \Phi_{12} \leftrightarrow \Phi_{21}.
\end{equation}

\section{Walker metric in Newman--Penrose formalism}

The spin coefficients of the Walker metrics \eref{eq:Walker_metric} in the frame \eref{eq:Walker:frame}
can be obtained by expressing the directional derivatives of the frame vectors
\begin{eqnarray}
  \D \ell^a = \D n^a = \D m^a = \D \tilde{m}^a = 0, \quad
  \T \ell^a = \T n^a = \T m^a = \T \tilde{m}^a = 0, \label{eq:Walker:derframevectors:first} \\
  \D' \ell^a = - (\varepsilon' + \tilde\varepsilon') \ell^a + \tau \tilde{m}^a, \quad
  \D' n^a = (\varepsilon' + \tilde\varepsilon') n^a - \kappa' m^a - \tilde\kappa' \tilde{m}^a, \\
  \D' \tilde{m}^a = (\varepsilon' - \tilde\varepsilon') \tilde{m}^a - \kappa' \ell^a, \quad
  \D' m^a = - (\varepsilon' - \tilde\varepsilon') m^a + \tau n^a - \tilde\kappa' \ell^a, \\
  \delta \ell^a = (\alpha' + \tilde\alpha) \ell^a + \sigma \tilde{m}^a, \quad
  \delta n^a = - (\alpha' + \tilde\alpha) n^a + \rho' m^a + \tilde\sigma' \tilde{m}^a, \\
  \delta \tilde{m}^a = - (\alpha' - \tilde\alpha) \tilde{m}^a + \rho' \ell^a, \quad
  \delta m^a = (\alpha' - \tilde\alpha) m^a + \sigma n^a + \tilde\sigma' \ell^a. \label{eq:Walker:derframevectors:last}
\end{eqnarray}
Therefore, the commutators of the directional derivatives for the Walker metrics reduce to
\begin{eqnarray}
  &[\D, \T] = 0, \quad
  [\D, \D'] = (\varepsilon' + \tilde\varepsilon') \D - \tau \T, \quad
  [\T, \D'] = \kappa' \D - (\varepsilon' - \tilde\varepsilon') \T, \label{eq:Walker:commutators:1} \\
  &[\D, \delta] = - (\alpha' + \tilde\alpha) \D - \sigma \T, \quad
  [\T, \delta] = - \rho' \D + (\alpha' + \tilde\alpha) \T, \label{eq:Walker:commutators:2} \\
  &[\delta, \D'] = \tilde\kappa' \D - (\tau + \tilde\alpha + \alpha') \D' + \tilde\sigma' \T + (\rho' - \tilde\varepsilon' + \varepsilon') \delta,
\end{eqnarray}

The components of the trace-free Ricci tensor \eref{eq:Walker:Ricci:first}--\eref{eq:Walker:Ricci:last}
and the Weyl tensor \eref{eq:Walker:Weyl:first}--\eref{eq:Walker:Weyl:last}
can be rewritten using \eref{eq:Walker:spincoeff:first}--\eref{eq:Walker:spincoeff:last}
in terms of the metric functions $A$, $B$, $C$, as
\begin{eqnarray}
  \Phi_{00} = \Phi_{10} = \Phi_{20} = 0, \quad
  \Phi_{11} = \frac{1}{4} A_{,vv} - \frac{1}{4} B_{,VV}, \label{eq:Walker:Phi00-11} \\
  \Phi_{01} = - \frac{1}{2} B_{,vV} - \frac{1}{4} C_{,vv}, \quad
  \Phi_{21} = \frac{1}{2} A_{,vV} + \frac{1}{4} C_{,VV}, \label{eq:Walker:Phi01,21} \\
  \Phi_{22} = \frac{1}{2} (2B A_{,V} + C A_{,v} - A C_{,v})_{,V} - A_{,UV} + \frac{1}{2} C_{,uV} - \frac{1}{8} (C^2)_{,VV}, \\
  \Phi_{02} = \frac{1}{2} (2A B_{,v} + C B_{,V} - B C_{,V})_{,v} - B_{,uv} + \frac{1}{2} C_{,vU} - \frac{1}{8} (C^2)_{,vv}, \\
  \Phi_{12} = \frac{1}{4} A (2 B_{,vV} - C_{,vv}) + \frac{1}{4} B (2 A_{,vV} - C_{,VV}) + \frac{1}{4} C (A_{,vv} + B_{,VV}) \\
    \qquad + A_{,V} B_{,v} - \frac{1}{2} A_{,vU} - \frac{1}{2} B_{,uV} + \frac{1}{4} C_{,uv} + \frac{1}{4} C_{,UV} - \frac{1}{8} (C^2)_{,vV} \label{eq:Walker:Phi12}
\end{eqnarray}
and
\begin{eqnarray}
  \Psi_0 = B_{,vv}, \quad
  \Psi_1 = \frac{1}{2} B_{,vV} - \frac{1}{4} C_{,vv}, \quad
  \Psi_2 = \frac{A_{,vv} + B_{,VV} - 2 C_{,vV}}{6}, \label{eq:Walker:Psi0-2} \\
  \Psi_3 = \frac{1}{2} A_{,vV} - \frac{1}{4} C_{,VV}, \quad
  \Psi_4 = A_{,VV}, \label{eq:Walker:Psi3,4} \\
  \tilde\Psi_0 = \tilde\Psi_1 = 0, \quad
  \tilde\Psi_2 = \frac{A_{,vv} + B_{,VV} + C_{,vV}}{6} = \frac{R}{12}, \label{eq:Walker:tildePsi0-2} \\
  \tilde\Psi_3 = -\frac{1}{4} \Big( A (C_{,vv} + 2 B_{,vV}) - B (C_{,VV} + 2 A_{,vV}) - C (A_{,vv} - B_{VV}) \nonumber \\
  \qquad + 2 A_{,vU} - 2 B_{,uV} - C_{uv} + C_{,UV} \Big), \\
  \tilde\Psi_4 = - A^2 B_{,vv} - B^2 A_{,VV} - \frac{1}{4} C^2 (A_{,vv} + B_{,VV} - C_{,vV}) + AB C_{,vV} \nonumber \\
    \qquad + A (B_{,v} C_{,V} - B_{,V} C_{,v} + 2 B_{,uv} - C_{,vU}) - \frac{1}{2} AC (2B_{,vV} - C_{,vv}) \nonumber \\
    \qquad + B (A_{,V} C_{,v} - A_{,v} C_{,V} + 2 A_{,UV} - C_{,uV}) - \frac{1}{2} BC (2A_{,vV} - C_{,VV}) \nonumber \\
    \qquad - \frac{1}{2} C (2 A_{,V} B_{,v} - 2 A_{,v} B_{,V} - 2 A_{,vU} - 2 B_{,uV} + C_{,uv} + C_{,UV}) \nonumber \\
    \qquad + A_{,u} B_{,v} - A_{,v} B_{,u} - A_{,U} B_{,V} + A_{,V} B_{,U} - A_{,UU} - B_{,uu} \nonumber \\
    \qquad + A_{,v} C_{,U} - A_{,U} C_{,v} - B_{,u} C_{,V} + B_{,V} C_{,u} + C_{,uU}, \label{eq:Walker:tildePsi4}
\end{eqnarray}
respectively.

In the dual-balance scalar approach employed in \Sref{section:EinsteinCSI} for determining
the possible contributions of the Weyl tensor and its covariant derivatives to rank-2 tensors,
it is crucial to express $\D$ and $\T$ derivatives of the spin coefficients and the Weyl tensor frame components.
Straightforwardly from the Ricci identities \cite{Law2008}
for Einstein Walker metrics, we get
\begin{eqnarray}
  \D\sigma = - \Psi_0, \quad
  \D\alpha' = - \D\tau = \Psi_1, \quad
  \D\kappa' = - \Psi_3, \quad
  \D\tilde\sigma' = \D\tilde\alpha = 0, \quad
  \label{eq:Walker:EinsteinASD:derspincoeff:first} \\
  \D\rho' = \Psi_2 - \frac{\lambda}{3}, \quad
  \D\varepsilon' = - \Psi_2 - \frac{\lambda}{6}, \quad
  \D\tilde\varepsilon' = - \frac{\lambda}{2}, \quad
  \D\tilde\kappa' =  - \tilde\Psi_3, \\
  \T\sigma = - \Psi_1, \quad
  \T\rho' = - \T\varepsilon' = \Psi_3, \quad
  \T\kappa' = - \Psi_4, \quad
  \T\tilde\kappa' = \T\tilde\varepsilon' = 0, \\
  \T\tau = - \Psi_2 + \frac{\lambda}{3}, \quad
  \T\alpha' = \Psi_2 + \frac{\lambda}{6}, \quad
  \T\tilde\alpha = - \frac{\lambda}{2}, \quad
  \T\tilde\sigma' = - \tilde\Psi_3.
  \label{eq:Walker:EinsteinASD:derspincoeff:last}
\end{eqnarray}
Applying tilde and prime on the Bianchi identities \eref{eq:Bianchi:1}, \eref{eq:Bianchi:2}
and tilde on the Bianchi identities \eref{eq:Bianchi:3}, \eref{eq:Bianchi:4}, respectively,
it immediately implies
\begin{eqnarray}
  \D\tilde\Psi_3 = \T\tilde\Psi_3 = 0.
  \label{eq:Walker:EinsteinASD:DtildePsi3} \\
  \D\tilde\Psi_4 - \delta\tilde\Psi_3 = - 3 \tilde\sigma' \tilde\Psi_2 + 2 \tilde\alpha \tilde\Psi_3, \\
  \T\tilde\Psi_4 - \D'\tilde\Psi_3 = 3 \tilde\kappa' \tilde\Psi_2 - 2 \tilde\varepsilon' \tilde\Psi_3.
  \label{eq:Walker:EinsteinASD:TtildePsi4}
\end{eqnarray}
Lastly, using \eref{eq:Walker:EinsteinASD:derspincoeff:first}--\eref{eq:Walker:EinsteinASD:TtildePsi4},
$\tilde\Psi_2 = \textnormal{const}$,
and the commutators \eref{eq:Walker:commutators:1}, \eref{eq:Walker:commutators:2}, it follows that
\begin{equation}
  \D^2\tilde\Psi_4 = \D\T\tilde\Psi_4 = \T^2\tilde\Psi_4 = 0.
  \label{eq:Walker:EinsteinASD:DtildePsi4}
\end{equation}

\section{Properties of $k$-dual-balanced scalars and tensors}
\label{appendix:balanced}

\begin{lemma}
  For a $k$-dual-balanced scalar $\eta$, a scalar $\mu\eta$ is $k$-dual-balanced if
  $b_\mu \equiv \bdiff(\mu) \ge 0$ and $\calD^p \mu = 0$ for $p \le b_\mu + 1$.
  \label{lemma:product}
\end{lemma}
\begin{proof}
  Since $\eta$ is $k$-dual-balanced, $b_\eta \equiv \bdiff(\eta) > k$.
  For $\mu\eta$ to be $k$-dual-balanced, necessarily $b_{\mu\eta} \equiv \bdiff(\mu\eta) = b_\mu + b_\eta > k$,
  which holds for $b_\mu$ non-negative.

  From the definition, $\calD^q \eta = 0$ for $q \ge b_\eta - k$
  and assume that there exists $p$ such that $\calD^p \mu = 0$.
  Now, $\calD^{b_{\mu\eta} - k}(\mu\eta) = \calD^{b_\mu + b_\eta - k}(\mu\eta) = 0$
  if $b_\mu + b_\eta - k \ge (p - 1) + (b_\eta - k - 1) + 1 = p + b_\eta - k - 1$
  and therefore $p \le b_\mu + 1$. 
\end{proof}

\begin{lemma}
  For a $k$-dual-balanced scalar $\eta$, scalars $\D\eta$ and $\T\eta$ are $k$-dual-balanced.
  \label{lemma:D&T}
\end{lemma}
\begin{proof}
  As follows from the definitions of the directional derivatives $\D$ and $\T$ \eref{eq:frame:directderiv},
  the application of $\calD$ on $\eta$ decreases the boost weight difference,
  i.e.\ $b_{\calD\eta} \equiv \bdiff(\calD\eta) = b_\eta - 1$, where $b_\eta \equiv \bdiff(\eta) > k$
  since $\eta$ is $k$-dual-balanced.
  If $b_{\calD\eta} \le k$ implying $b_\eta = k + 1$, then according to the first condition
  \eref{def:balanced:cond1} $\calD\eta$ has to vanish, which holds due to $\calD^{b_\eta - k}\eta = 0$.

  Otherwise $b_{\calD\eta} > k$ and the second condition \eref{def:balanced:cond2} for $\calD\eta$
  is met since $\calD^{(b_\eta - 1) - k}(\calD\eta) = \calD^{b_\eta - k}\eta = 0$.
\end{proof}

\begin{lemma}
  In Einstein Walker spaces with a constant self-dual part of the Weyl tensor,
  for a $k$-dual-balanced scalar $\eta$, scalars $\D'\eta$ and $\delta\eta$ are $k$-dual-balanced.
  \label{lemma:D'&delta}
\end{lemma}
\begin{proof}
  The derivatives $\D'$ and $\delta$ \eref{eq:frame:directderiv} increase the boost weight difference,
  i.e.\ $\bdiff(\D'\eta) = \bdiff(\delta\eta) = b_\eta + 1 > k$. Thus it remains to show that
  $\calD^{(b_\eta + 1) - k}(\D'\eta)$ and $\calD^{(b_\eta + 1) - k}(\delta\eta)$ vanish.
  In the former case, using the commutators \eref{eq:Walker:commutators:1} we obtain
  \begin{equation}
    \calD^{b_\eta - k}\D\D'\eta = \calD^{b_\eta - k}\D'\D\eta
      + \calD^{b_\eta - k}[(\varepsilon' + \tilde\varepsilon')\D\eta]
      - \calD^{b_\eta - k}[\tau\T\eta]
    \label{eq:lemmaD'}
  \end{equation}
  or
  \begin{equation}
    \calD^{b_\eta - k}\T\D'\eta = \calD^{b_\eta - k}\D'\T\eta
      + \calD^{b_\eta - k}[\kappa'\D\eta]
      - \calD^{b_\eta - k}[(\varepsilon' - \tilde\varepsilon')\T\eta].
    \label{eq:lemmadelta}
  \end{equation}
  The application of the Leibniz rule for the last two terms in \eref{eq:lemmaD'} and \eref{eq:lemmadelta} yields terms of the form
  \begin{equation*}
    \calD^p\tau \calD^q\eta, \quad
    \calD^p\kappa' \calD^q\eta, \quad
    \calD^p\varepsilon'\calD^q\eta, \quad
    \calD^p\tilde\varepsilon' \calD^q\eta
  \end{equation*}
  with $p + q = b_\eta - k + 1$. All these terms vanish due to the assumption that $\calD^{b_\eta - k}\eta = 0$
  or since $\calD^2\kappa' = \calD^2\tau = \calD^2\varepsilon' = \calD^2\tilde\varepsilon' = 0$
  as it follows for Einstein Walker metrics with constant self-dual components of the Weyl tensor
  from \eref{eq:Walker:EinsteinASD:derspincoeff:first}--\eref{eq:Walker:EinsteinASD:derspincoeff:last}.
  Using the commutators and the same arguments repeatedly, we finally get 
  $\calD^{b_\eta - k + 1} \D'\eta = \D'\calD^{b_\eta - k + 1}\eta = 0$. 

  Analogously, one can show that $\delta\eta$ is $k$-dual-balanced using the commutators
  \eref{eq:Walker:commutators:2} and $\calD^2\sigma = \calD^2\rho' = \calD^2\alpha' = \calD^2\tilde\alpha = 0$
  which follows from
  \eref{eq:Walker:EinsteinASD:derspincoeff:first}--\eref{eq:Walker:EinsteinASD:derspincoeff:last}
  for Einstein Walker metrics with constant $\Psi_i$.
\end{proof}

\begin{lemma}
  For Einstein Walker metrics with a constant self-dual part of the Weyl tensor,
  the covariant derivative of a $k$-dual-balanced tensor is a $k$-dual-balanced tensor.
  \label{lemma:covd}
\end{lemma}
\begin{proof}
  Taking the covariant derivative $\nabla_a = n_a \D + \ell_a \D' - m_a \T - \tilde m_a \delta$
  of a tensor, there appear terms involving directional derivatives of the tensor frame components
  \begin{equation}
    \D\eta, \D'\eta, \T\eta, \delta\eta
    \label{eq:lemmacovd:dereta}
  \end{equation}
  and directional derivatives of the frame vectors
  \eref{eq:Walker:derframevectors:first}--\eref{eq:Walker:derframevectors:last}
  leading to
  \begin{equation}
    \sigma \eta, \tau \eta, \kappa' \eta, \tilde\kappa' \eta, \rho' \eta, \tilde\sigma' \eta, \varepsilon' \eta,
    \tilde\varepsilon' \eta, \tilde\alpha \eta, \alpha' \eta.
    \label{eq:lemmacovd:mueta}
  \end{equation}
  For a $k$-dual-balanced component $\eta$, the terms \eref{eq:lemmacovd:dereta} are $k$-dual-balanced
  as follows from Lemmas \ref{lemma:D&T} and \ref{lemma:D'&delta}.
  For Einstein Walker metrics with constant self-dual components of the Weyl tensor,
  the ten non-trivial spin coefficients split to two groups
  according to their boost weight differences
  \begin{eqnarray}
    &\bdiff(\mu) = 1: \quad \mu \in \{\sigma, \tau, \kappa', \rho', \varepsilon', \tilde\varepsilon', \tilde\alpha, \alpha'\} \label{eq:Walker:EinsteinASD:spincoeff:-1db} \\
    &\bdiff(\mu) = 3: \quad \mu \in \{\tilde\kappa', \tilde\sigma'\} \label{eq:Walker:EinsteinASD:spincoeff:1db}
  \end{eqnarray}
  and $\calD^2 \mu = 0$ in both cases as can be seen directly from 
  \eref{eq:Walker:EinsteinASD:derspincoeff:first}--\eref{eq:Walker:EinsteinASD:DtildePsi3}.
  Then Lemma \ref{lemma:product} implies that the terms \eref{eq:lemmacovd:mueta} are $k$-dual-balanced,
  which finishes the proof.
\end{proof}

\section{Discrete frame transformations}
\label{appendix:tranfs}

The frame transformation \eref{eq:pair_interchange2} interchanging the pairs of the
frame vectors provides several implications for rank-2 tensors.

\begin{lemma}
  The b.w.\ (0,0) part of any symmetric rank-2 tensor polynomial in $C^{(b,b)}$
  is proportional to the metric.
  \label{lemma:C(0,0)^k}
\end{lemma}
\begin{proof}
  The b.w.\ (0,0) part of a symmetric rank-2 tensor takes the form
  $S^{(0,0)}_{ab} = S_{01} \, \ell_{(a} n_{b)} + S_{23} \, m_{(a} \tilde m_{b)}$,
  which need not to be proportional to the metric in general. Since $C^{(b,b)}$
  are invariant under \eref{eq:pair_interchange2}, any symmetric rank-2 tensor
  $S^{(0,0)}$ polynomial in $C^{(b,b)}$ is thus invariant as well, implying
  $S_{01} \mapsto - S_{23}$. On the other hand, the components of such $S^{(0,0)}$
  are constructed only from b.w.\ (0,0) blocks
  \begin{equation}
    \Psi_2, \tilde\Psi_2, \Psi_1\Psi_3, \Psi_0\Psi_4, \Psi_1^2\Psi_4, \Psi_0\Psi_3^2
    \label{eq:bw(0,0)_blocks}
  \end{equation}
  and their arbitrary combinations. All these blocks are invariant under
  \eref{eq:pair_interchange2} and therefore $S_{01} \mapsto S_{01}$. It immediately
  follows that $S_{01} = - S_{23}$ and thus $S^{(0,0)}$ is proportional to the metric.
\end{proof}

\begin{lemma}
  The b.w.\ $(-1,-1)$ and $(1,1)$ parts of a rank-2 tensor with components changing
  the sign under \eref{eq:pair_interchange2} are anti-symmetric.
  \label{lemma:T(1,1)}
\end{lemma}
\begin{proof}
  The b.w.\ $(-1,-1)$ and $(1,1)$ parts of a rank-2 tensor $T_{ab}$ can be written as
  $T^{(1,1)}_{ab} = T_{02} \, n_a \tilde m_b + T_{20} \, \tilde m_a n_b$ and 
  $T^{(-1,-1)}_{ab} = T_{13} \, \ell_a m_b + T_{31} \, m_a \ell_b$, respectively,
  where $T_{02} = - T_{ab} \, \ell^a m^b$, $T_{20} = - T_{ab} \, m^a \ell^b$, etc.
  Obviously, $T_{02} \leftrightarrow T_{20}$ and $T_{13} \leftrightarrow T_{31}$
  under \eref{eq:pair_interchange2}.
  On the other hand, $T_{02} \mapsto - T_{02}$ and $T_{13} \mapsto - T_{13}$ by assumption.
  Hence, $T_{02} = - T_{20}$ and $T_{13} = - T_{31}$.
\end{proof}

In the other case of the discrete frame transformation \eref{eq:pair_interchange},
one can analogously prove:

\begin{lemma}
  The b.w.\ $(-1,1)$ part of a rank-2 tensor with components changing the sign
  under \eref{eq:pair_interchange} is anti-symmetric.
  \label{lemma:T(-1,1)}
\end{lemma}

\begin{lemma}
  The b.w.\ $(-2,0)$ and $(0,2)$ parts of a rank-2 tensor $T_{ab}$ vanish
  if the components $T_{11} = T_{ab} \, n^a n^b$ and $T_{22} = T_{ab} \, m^a m^b$
  transform as $T_{11} \leftrightarrow - T_{22}$ under \eref{eq:pair_interchange}.
  \label{lemma:T(0,2)}
\end{lemma}

\section{Einstein Walker metrics with a constant self-dual part of the Weyl tensor}
\label{appendix:DC}

In this Appendix, we first show that higher order directional derivatives
of spin coefficients can be rewritten as a sum of products of first derivatives
of spin coefficients.

\begin{lemma}
  For any $k$-dual-balanced spin coefficient $\mu$ of Einstein Walker metrics
  with a constant self-dual part of the Weyl tensor and a positive integer $n$,
  it holds that $\{\calD^{n} \calD'^{n-1}\} \mu = \sum (\calD\xi)^{n-1} \calD\mu$,
  where the sum may be empty.
  \label{lemma:DD'mu}
\end{lemma}
\begin{proof}
  We prove the lemma by induction. But first, note that
  \begin{equation}
    \{\calD^{n+1} \calD'^{n-1}\} \mu = 0
    \label{eq:lemma:DD'mu}
  \end{equation}
  since Lemmas \ref{lemma:D&T} and \ref{lemma:D'&delta} imply that the left hand side is $k$-dual-balanced
  and $\bdiff(\{\calD^{n+1} \calD'^{n-1}\} \mu) = \bdiff(\mu) - 2 \le k$ is met for both groups of spin coefficients
  \eref{eq:Walker:EinsteinASD:spincoeff:-1db} and \eref{eq:Walker:EinsteinASD:spincoeff:1db}
  of Einstein Walker metrics.
  It immediately follows that terms $\{\calD^n \calD'^{n-1}\} \mu$ starting with $\calD'$ vanish
  since the trailing part is of the form \eref{eq:lemma:DD'mu}.

  The step $n = 1$ is trivial.
  Next, let us assume the statement holds for $n = l$
  and show that it then also holds for $n = l + 1$ by commuting the outermost $\calD'$ to the left.
  Using \eref{eq:lemma:DD'mu}, $\calD^2\xi = 0$, the commutators
  \eref{eq:Walker:commutators:1}--\eref{eq:Walker:commutators:2}
  \begin{equation}
    [\calD,\calD'] = \sum \xi \calD
    \label{eq:lemma:DD'mu:commutators}
  \end{equation}
  and substituting the assumption, it follows that
  \begin{eqnarray}
    \fl
    \{\calD^{l+1} \calD'^l\} \mu &= \calD^m \calD' \{ \calD^{l+1-m} \calD'^{l-1} \}\mu \nonumber \\
    \fl
    &= \calD^{m-1} \calD' \calD \{ \calD^{l+1-m} \calD'^{l-1} \}\mu + \calD^{m-1} \left( \sum \xi \calD \{ \calD^{l+1-m} \calD'^{l-1} \}\mu \right) \nonumber \\
    \fl
    &= \calD^{m-1} \calD' \calD \{ \calD^{l+1-m} \calD'^{l-1} \}\mu + \sum \xi \{ \calD^{l+1} \calD'^{l-1} \}\mu
    + \sum \calD\xi \{ \calD^l \calD'^{l-1} \}\mu \nonumber \\
    \fl
    &= \calD^{m-1} \calD' \calD \{ \calD^{l+1-m} \calD'^{l-1} \}\mu + \sum (\calD \xi)^l \calD\mu.
  \end{eqnarray}
  Commuting $\calD'$ repeatedly, we finally get
  \begin{eqnarray}
    \fl
    \{\calD^{l+1} \calD'^l\} \mu &= \calD \calD' \calD^{m-1} \{ \calD^{l+1-m} \calD'^{l-1} \}\mu
    + \sum (\calD \xi)^l \calD\mu \nonumber \\
    \fl
    &= \calD' \calD^m \{ \calD^{l+1-m} \calD'^{l-1} \}\mu + \sum \xi \calD^m \{ \calD^{l+1-m} \calD'^{l-1} \}\mu
    + \sum (\calD \xi)^l \calD\mu \nonumber \\
    \fl
    &= \calD' \{ \calD^{l+1} \calD'^{l-1} \}\mu + \sum \xi \{ \calD^{l+1} \calD'^{l-1} \}\mu
    + \sum (\calD \xi)^l \calD\mu
  \end{eqnarray}
  and thus
  \begin{equation}
    \{\calD^{l+1} \calD'^l\} \mu = \sum (\calD \xi)^l \calD\mu
  \end{equation}  
\end{proof}

\begin{lemma}
  For $\bdiff = 1$ spin coefficients $\xi$ of Einstein Walker metrics with a constant self-dual
  part of the Weyl tensor and a positive integer $n$, it holds that
  $\{\calD^{n-1} \calD'^{n-1}\} \xi = \sum (\calD\xi)^{n-1} \xi$, where the sum
  may be empty.
  \label{lemma:DD'xi}
\end{lemma}
\begin{proof}
  The case $n = 1$ is trivial. Similarly as in Lemma \ref{lemma:DD'mu}, we prove
  the induction step $n = l + 1$ by commuting the outermost $\calD'$ to the left.
  Using \eref{eq:lemma:DD'mu}, $\calD^2\xi = 0$, the commutators \eref{eq:lemma:DD'mu:commutators}
  and substituting the assumption that the statement holds for $n = l$, we get
  \begin{eqnarray}
    \fl
    \{\calD^{l} \calD'^l\} \xi &= \calD^m \calD' \{ \calD^{l-m} \calD'^{l-1} \}\xi \nonumber \\
    \fl
    &= \calD^{m-1} \calD' \calD \{ \calD^{l-m} \calD'^{l-1} \}\xi + \calD^{m-1} \left( \sum \xi \calD \{ \calD^{l-m} \calD'^{l-1} \}\xi \right) \nonumber \\
    \fl
    &= \calD^{m-1} \calD' \calD \{ \calD^{l-m} \calD'^{l-1} \}\xi + \sum \xi \{ \calD^{l} \calD'^{l-1} \}\xi
    + \sum \calD\xi \{ \calD^{l-1} \calD'^{l-1} \}\xi \nonumber \\
    \fl
    &= \calD^{m-1} \calD' \calD \{ \calD^{l-m} \calD'^{l-1} \}\xi + \sum \xi (\calD \xi)^l.
  \end{eqnarray}
  Commuting $\calD'$ repeatedly leads to
  \begin{eqnarray}
    \{\calD^{l} \calD'^l\} \xi &= \calD \calD' \calD^{m-1} \{ \calD^{l-m} \calD'^{l-1} \}\xi
    + \sum \xi (\calD \xi)^l \nonumber \\
    &= \calD' \calD^m \{ \calD^{l-m} \calD'^{l-1} \}\xi + \sum \xi \calD^m \{ \calD^{l-m} \calD'^{l-1} \}\xi
    + \xi \sum (\calD \xi)^l \nonumber \\
    &= \calD' \{ \calD^{l} \calD'^{l-1} \}\xi + \sum \xi \{ \calD^{l} \calD'^{l-1} \}\xi
    + \sum \xi (\calD \xi)^l.
  \end{eqnarray}
  Finally, employing Lemma \ref{lemma:DD'mu} and using the fact that $\calD\xi$ is constant
  for Einstein Walker metrics with constant self-dual components of the Weyl tensor, we obtain
  \begin{equation}
    \{\calD^{l} \calD'^l\} \xi = \sum \xi (\calD \xi)^l.
  \end{equation}  
\end{proof}

As a consequence of these results, it turns out that relevant components of covariant derivatives
of anti-self dual part of the Weyl tensor can be written as a sum of products of Weyl tensor
components and first derivatives of spin coefficients:

\begin{lemma}
  For Einstein Walker metrics with a constant self-dual part of the Weyl tensor,
  $\bdiff = 2$ components of $\nabla^{(2n)}C_{abcd}^{(0,0)^-}$ and $\nabla^{(2n)}C_{abcd}^{(-1,1)}$
  take the form $\tilde\Psi_2 \tilde\Psi_3 \sum(\calD\xi)^{n-1}$
  and $\tilde\Psi_3 \sum(\calD\xi)^n$, respectively.
  \label{lemma:DC}
\end{lemma}
\begin{proof}
  Applying covariant derivatives on $C^{(-1,1)}$, there appear frame components
  consisting of $\tilde\Psi_3$ or its directional derivatives, spin coefficients
  and subsequently their directional derivatives. We already know $\tilde\Psi_3$
  is 1-dual-balanced for Einstein Walker metrics, therefore, $\tilde\Psi_3$ and
  its directional derivatives are of $\bdiff \ge 2$. Moreover, the only non-vanishing
  $\bdiff = 2$ term $\{\calD^k\calD'^k\}\tilde\Psi_3$ is just $\tilde\Psi_3$ when $k = 0$
  as can be shown commuting the innermost $\calD$ to the right and using $\calD\tilde\Psi_3 = 0$.

  Since $\tilde\sigma'$, $\tilde\kappa'$ are 1-dual-balanced and $\sigma$, $\kappa'$,
  $\tau$, $\tilde\alpha$, $\alpha'$, $\rho'$, $\varepsilon'$, $\tilde\varepsilon'$
  are $-1$-dual-balanced for Einstein Walker metrics with a constant self-dual
  part of the Weyl tensor, it is not possible to construct any term with negative
  boost weight difference from spin coefficients and moreover $\bdiff = 0$ terms
  involve only $\xi$.
  Hence, non-vanishing $\bdiff = 2$ components of $\nabla^{(k)} C^{(-1,1)}$
  are made of only $\tilde\Psi_3$ and any number of $\{\calD^l\calD'^{l-1}\}\xi$
  terms, where each such a term appears due to the application of $2l$ covariant
  derivatives. Then Lemma \ref{lemma:DD'mu} implies that $\bdiff = 2$ terms of
  $\nabla^{(2n)} C^{(-1,1)}$ can be written as $\tilde\Psi_3 \sum (\calD\xi)^n$.

  Similar arguments also apply to $C^{(0,0)^-}$. But now, first covariant derivative
  generates only $\bdiff = 3$ components $\tilde\Psi_2 \tilde\sigma'$ and
  $\tilde\Psi_2 \tilde\kappa'$, see \eref{eq:Einstein:nablaC(0,0)-}. Therefore,
  $\bdiff = 2$ components can be constructed only from $\{\calD^k\calD'^{k-1}\}\tilde\sigma'$
  or $\{\calD^l\calD'^{l-1}\}\tilde\kappa'$ and any number of $\bdiff = 0$ terms
  $\{\calD^m\calD'^{m-1}\}\xi$. As follows from Lemma \ref{lemma:DD'mu}, $\bdiff = 2$
  terms of $\nabla^{(2n)} C^{(0,0)^-}$ are thus of the form
  $\tilde\Psi_2 \calD\tilde\sigma' \sum (\calD\xi)^{n-1}$ or
  $\tilde\Psi_2 \calD\tilde\kappa' \sum (\calD\xi)^{n-1}$.
\end{proof}

\section*{References}
\bibliographystyle{plain}

\begin{thebibliography}{1}

\bibitem{Clifton:2011} 
Clifton T, Ferreira P G , Padilla A and Skordis C 2012 Modified Gravity and Cosmology
{\it  Phys.\ Rept.}  {\bf 513} 1--189 (arXiv:1106.2476)

\bibitem{Coleyetal08}
Coley A A, Gibbons G W, Hervik S and Pope C N 2008 Metrics with vanishing quantum corrections
\CQG {\bf 25} 145017 (arXiv:0803.2438)

\bibitem{Bleecker} 
Bleecker D D 1979 Critical Riemannian manifolds {\it J. Di. Geo.} {\bf 14} 599--608

\bibitem{HorSte90}
Horowitz G T and Steif A R 1990 Spacetime singularities in string theory \PRL {\bf 64} 260--263

\bibitem{HerPraPra14}
Hervik S, Pravda V and Pravdov\'a A 2014 Type {III} and {N} universal spacetimes \CQG {\bf 31} 215005 (arXiv:1311.0234)
	
\bibitem{univII}
Hervik S, M\'alek T, Pravda V and Pravdov\'a A 2015 Type II universal spacetimes \CQG {\bf 32} 245012 (arXiv:1503.08448)

\bibitem{HerPraPra17}
Hervik S, Pravda V and Pravdov\'a A 2017 Universal spacetimes in four dimensions \JHEP {\bf 10} 28 arXiv:1707.00264

\bibitem{Kundt}
Kundt W 1961 The plane-fronted gravitational waves {\it Z. Physik} {\bf 163} 77

\bibitem{Coleyetal2003}
Coley A et al 2003 Generalizations of pp-wave spacetimes in higher dimensions \PR D {\bf 67} 104020 (arXiv:gr-qc/0212063)

\bibitem{ColHerPel06}
Coley A, Hervik S and Pelavas N 2006 On spacetimes with constant scalar invariants \CQG {\bf 23} 3053--3074 (arXiv:gr-qc/0509113)	

\bibitem{PodolskyZofka2009}
Podolsk\'y J and \v{Z}ofka M 2009 General Kundt spacetimes in higher dimensions \CQG {\bf 26} 105008 (arXiv:0812.4928)

\bibitem{wick} 
Helleland C and Hervik S 2018 Wick rotations and real GIT {\it J.\ Geom.\ Phys.}  {\bf 123} 343 (arXiv:1703.04576)

\bibitem{twist}
Penrose R and Rindler W 1986 {\it Twistors and Space-time} (Cambridge Uni. Press) volume 1 \& 2  

\bibitem{twist2}
Dunajski M 2009 \JPA {\bf 42} 404004 (arXiv:0902.0274)

\bibitem{twist3}
Dunajski M 2002 {\it Proc.\ Roy.\ Soc.\ Lond.} {\bf A458} 1205--1222  

\bibitem{twist4}
Dunajski M and Tod P 2010 {\it Math.\ Proc.\ Cam.\ Phil.\ Soc.} {\bf 148} 485 (arXiv:0901.2261)

\bibitem{OoguriVafa}
Ooguri H and Vafa C 1990 Self-duality and N = 2 string magic {\it Modern Phys.\ Lett.} {\bf A5} 1389--1398

\bibitem{pseudoVSI2}
Hervik S 2012 Pseudo-Riemannian VSI spaces II \CQG {\bf 29} 095011 (arXiv:1504.1616)

\bibitem{Law2008}
Law P R 2009 Spin coefficients for four-dimensional neutral metrics, and null geometry
{\it J. Geometry Phys.} {\bf 59} 1087--1126 (arXiv:0802.1761)

\bibitem{Walker1950}
Walker A G 1950 Canonical form for a riemannian space with a parallel field of null planes {\it Quart. J. Math.} {\bf 1} 69--79

\bibitem{PravdaPravdovaColeyMilson2002}
Pravda V, Pravdov\'a A, Coley A and Milson R 2002 All space-times with vanishing curvature invariants \CQG {\bf 19} 6213--6236
(arXiv:gr-qc/0209024)

\end{thebibliography}

\end{document}